\pdfoutput=1
\PassOptionsToPackage{unicode}{hyperref}
\PassOptionsToPackage{naturalnames}{hyperref}

\def\thm@space@setup{%
	\thm@preskip=2cm plus 1cm minus 2cm
	\thm@postskip=\thm@preskip 
}

\documentclass[11pt]{article}
\usepackage[margin=0.8in]{geometry}
\usepackage{amsthm}
\usepackage{amsfonts}
\usepackage{amssymb}
\usepackage{amsmath}
\usepackage{float}
\usepackage{caption}
\usepackage{subcaption}
\usepackage[newenum]{paralist}
\usepackage[usenames,dvipsnames]{xcolor}
\usepackage[T1]{fontenc}
\usepackage{verbatim}
\usepackage[bottom]{footmisc}
\usepackage{calligra}
\usepackage{setspace} 
\usepackage{boxedminipage}
\usepackage{multirow}
\usepackage{graphicx}
\usepackage{hyperref}
\usepackage{framed}
\usepackage{algorithm}
\usepackage{algorithmic}

\makeatletter
\usepackage{linegoal}

\makeatletter
\renewcommand\section{\@startsection {section}{1}{\z@}%
	{-2.2ex \@plus -1ex \@minus -.2ex}%
	{1.2ex \@plus.1ex}%
	{\normalfont\Large\bfseries}}
\makeatother

\makeatletter
\renewcommand\subsection{\@startsection {subsection}{1}{\z@}%
	{-2ex \@plus -1ex \@minus -.2ex}%
	{1ex \@plus.1ex}%
	{\normalfont\large\bfseries}}
\makeatother

\makeatletter
\newtheorem*{rep@theorem}{\rep@title}
\newcommand{\newreptheorem}[2]{%
	\newenvironment{rep#1}[1]{%
		\def\rep@title{#2 \ref{##1}}%
		\begin{rep@theorem}}%
		{\end{rep@theorem}}}
\makeatother

\newtheorem{theorem}{Theorem}[section]
\newtheorem{lemma}[theorem]{Lemma}
\newtheorem{axiom}[theorem]{Axiom}

\newreptheorem{theorem}{Theorem}
\newreptheorem{lemma}{Lemma}
\newreptheorem{corollary}{Corollary}
\hypersetup{colorlinks,linkcolor=blue,filecolor=blue,citecolor=blue, urlcolor=blue,pdfstartview=FitH}

\begin{document}
	\title{Hard Capacitated Set Cover and Uncapacitated Geometric Set Cover}
	\author{Rahil Sharma\\ \small{Department of Computer Science, The University of Iowa, Iowa City, IA 52242}\\ 
		\texttt{rahil-sharma@uiowa.edu}}
	\date{}
	\thispagestyle{empty}
	\pagenumbering{gobble}
	\maketitle
	
\section{Introduction}\label{sec:intro}
In the optimization problem, \textit{minimum set cover}, we are given a set system $(X,\cal{S})$, where $X$ is a ground set of $n$ elements
and $\cal S$ is a set of subsets of $X$. Each set $S \in \cal S$ has a cost $w(S) > 0$ associated with it.
Let $\cal R$ be a subset of $\cal S$. We say that $\cal R$ is a set cover of $X$, if each element in $X$ belongs to at least one set in $\cal R$. 
The objective is to find a set cover $\cal R$ that minimize the sum of the cost of all sets in $\cal R$. The minimum set cover problem is a classic
NP-hard problem \cite{zuckerman1993np}. In the example below,
$X = \{1, 2, 3, 4, 5, 6\}$ and $\cal{S}$ $= \{S_1, S_2, S_3, S_4\}$, where $S_1 = \{1, 2\}$, $S_2 = \{2, 3, 4\}$, $S_3 = \{4, 5, 6\}$, $S_4 = \{5, 6\}$. The costs of the 
four sets are ${1, 2, 5, 3}$ respectively.
\begin{figure}[H]
\centering
\includegraphics[scale = 0.46]{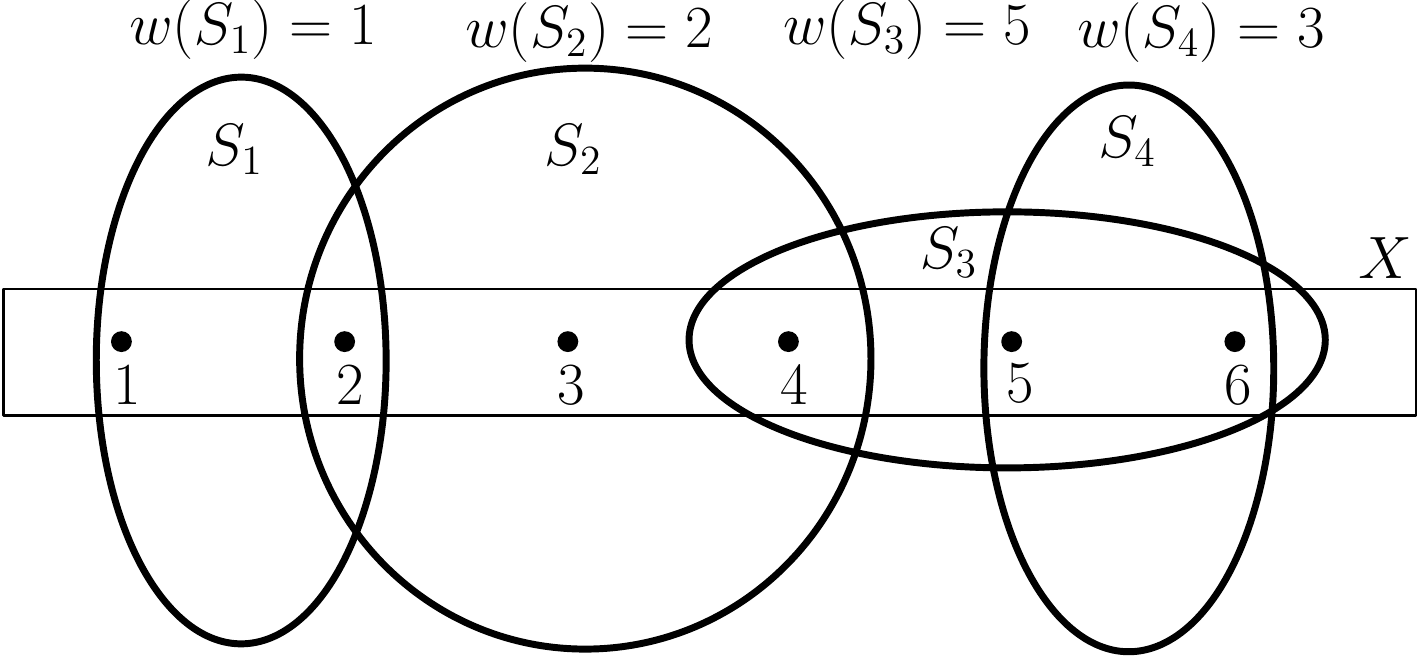}
\caption{\textit{$\cal R$ $= \{S_2, S_4\}$ is not a set cover of $X$, since element $1 \in X$ is not contained in any of the sets in $\cal R$. There are 2 possible set covers 
                 for this example, $\cal{R}$ $= \{S_1, S_2, S_3 \}$ and $\cal{R}$ $= \{S_1, S_2, S_4 \}$ of cost 8 and 6 respectively. Hence $\cal{R}$ $= \{S_1, S_2, S_4 \}$ 
                 is the minimum set cover here.}}
\end{figure}
We are interested in the capacitated version of minimum set cover problem, in which each set ${S} \in \cal S$ has a capacity $k(S)$ associated with it, such that
for each set $S \in \cal S$, at most $k(S)$ elements can be contained in $S$.
The capacitated covering problems are of 2 types. First, sets with \textit{soft capacities} where each set has unbounded number of copies that can 
be used to cover and the second, sets with \textit{hard capacities} where each set has a bound on the number of available copies.\par
Formally, in a \textbf{capacitated set cover problem with hard capacities} we are given a ground set of elements $X$ and a collection 
of its subset $\cal S$. Each $S\in \cal S$ has a positive integral capacity $k(S)$ and a non-negative cost $w(S)$ associated with it. Let $\cal R$ be a 
set of subsets of $\cal S$. Let $f: X \rightarrow \cal R$ be an assignment of elements in $X$ to a set in $\mathcal{R}$ such that for any $x\in X$, if 
(a) $f(x) = S$ then $x \in S$ (b) $|\{x\mid f(x) = S\}| \leq k(S)$ for all $S\in \cal R$. We call $\cal R$ a valid set cover if such $f$ exists. 
Cost of the solution $\cal R$ is the sum of the cost of all sets in $\cal R$.
The goal is to find a set cover $\cal R$ of minimum cost i.e. $\sum_{S\in \cal R} w(S)$.
The best known result for this problem was given by Wolsey's algorithm \cite{wolsey1982analysis}. It gave a set cover of cost $O(log \hspace{1mm} n)$ of the optimal solution. 
In the first of the two parts of this report, we show the same result using Wolsey's algorithm, but with a different and simpler analysis given in \cite{chuzhoy2002covering}.
We also make a key observation in the analysis given in \cite{chuzhoy2002covering} which allows us to apply the weighted set cover greedy algorithm's 
analysis \cite{chvatal1979greedy} to hard capacitated set cover problem. \par
One of the motivating applications is [4]; i.e. \textit{GMID (Glyco-molecule ID)} which is a chip-based technology that is used to generate 
fingerprints which uniquely identify glycomolecules. Each run of the experiment answers the question: For a given building block A, and for each member
B in a set S of building blocks, does the solution contain a molecule which contains both building blocks A and B? The size of the set S is restricted, 
because of the specific technology. Here the information is presented as a graph where the building blocks are its vertices, and an edge exists between two
vertices if the question regarding their connectivity is required. The device is able to answer $|S| = k$ questions at once if they share a 
common vertex. The problem of minimizing the number of GMID experiments needed to cover the required information graph, is precisely a capacitated 
vertex cover with hard capacities which is a special case of the set cover, with each element belongs to two sets \cite{guha2002capacitated}. It is easy to see the generalization 
of this special case to our problem where the information is provided in a multi-graph.\par
Another interesting application this problem has is in the area of \textit{cellular network coverage}. We are given $n$ users and a set $\cal S$ $= \{S_1,S_2,...,S_m\}$ 
of $m$ antennas. Each antenna has a capacity $k(S_i)$ and a positive cost of installation and maintenance $c(S_i)$, for all i where, $1 \leq i \leq m$; 
associated with it. The capacity determines the upper-bound on the number of users it can serve at a time, no matter how many users it covers. This is a real-life 
bandwidth issue, that needs to be resolved while setting up the network. The problem here is to find a minimum cost subset of antennas from 
$\cal S$, that serves all the n-users, without violating the capacity constraint. Provided that a feasible solution for the problem exists.
Here the coverage area of the antenna can be viewed as geometric objects like a disk centered at the antenna or a sector originating from an antenna. 
This drives our motivation behind exploring geometric set cover problems where the covering objects/sets are geometric bodies.
\begin{figure}[H]
\centering 
\includegraphics[scale = 0.45]{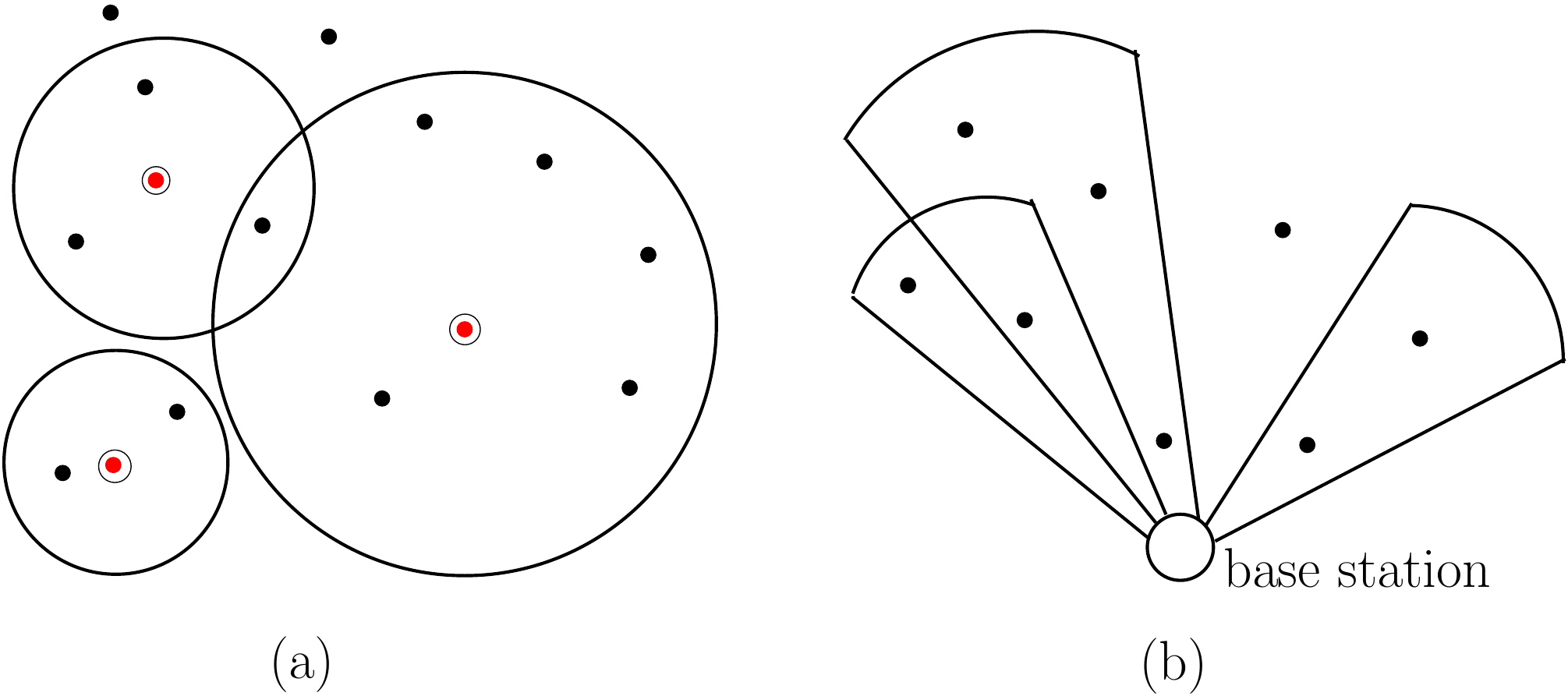}
\caption{\textit{(a) Antenna located at the center of the circle, which represents the coverage area (b) A single base station with 3 directional antennas on it and the coverage 
area resembles a sector}}
\end{figure}
\par

In the \textbf{geometric set cover problem}, we are given a ground set $X$ of $n$ points in $\Re^{2}$ and a set of objects $\cal S$ whose union contains all the points in $X$. 
The objects are geometric bodies like fat triangles, discs, rectangles, etc. The goal is to find the minimum cardinality subset $\cal R \subseteq \cal S$ such that every point 
in $X$ is contained in at least one object in $\cal R$. It is important to note that, there is another class of problems called as hitting set problem. 
In the related \textbf{hitting set problem}; the goal will be to select a minimum cardinality subfamily $X'\subseteq X$, such that every object in 
$\cal S$ contains at least one point in $X'$. \par
We can reduce the above two problems to the general set cover problem. The geometric set cover problem is reduced to the general set cover problem (see Figure 3(a)), by viewing the points in the 
plane as elements to be covered and objects as sets to cover them. The hitting set problem is reduced to general set cover problem (see Figure 3(b)), by viewing each point in the plane 
as a set and each object as an element to be covered. The geometric set cover problem is the dual problem for the geometric hitting set problem.
This is an \textit{unweighted} version of the general set cover problem, where each object in $\cal S$ has no cost/weight associated with it or they have uniform weight.

\begin{figure}[H]
\centering 
\includegraphics[scale = 0.58]{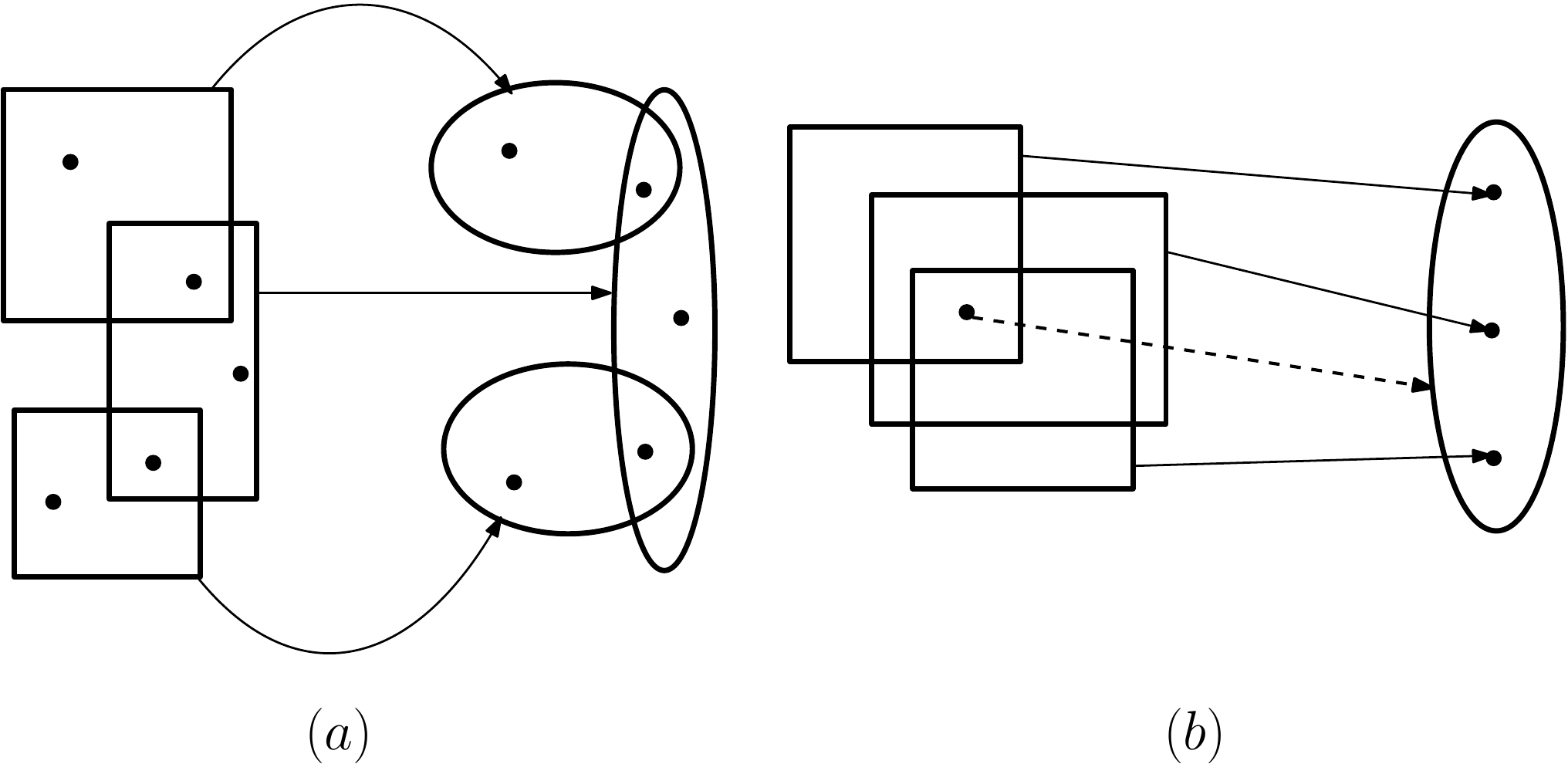}
\caption{\textit{(a) Geometric set cover instance to general set cover (b) Hitting set instance to general set cover}}
\end{figure}
\par

Let us define \textit{$\epsilon$ -net} since it is the central theme used for improving bounds on geometric set cover problems.
Let $X$ be a set of $n$ points in the plane and $\cal S$ be the set of axis parallel rectangles whose union covers $X$. For $\epsilon \in (0,1)$, the 
$\epsilon$ -net of $X$ is a set $Y$ of points in a plane, such that any rectangle $S \in \cal S$ that contains at least $\epsilon|X|$ points from $X$, contains a point from $Y$. 
One can even think of $\epsilon$ -nets as a cover for heavily covered points. 
The connection between the size of $\epsilon$ -nets and size of the hitting set is shown by Bronnimann and Goodrich \cite{bronnimann1995almost}, i.e. if for any point set
$\epsilon$ -net of size $O(1/\epsilon\hspace{1mm} h(1/\epsilon))$ is computed in polynomial time, then the corresponding geometric hitting set problem has a polynomial time 
$O(h(OPT))$ approximation, where $OPT$ is the size of the optimal cover. Haussler and Welzl \cite{haussler_epsilon-nets_1987} showed the existence 
of such nets of size $O(1/\epsilon\hspace{1mm} log\hspace{1mm}(1/\epsilon))$ for general geometric objects. 
Thus using the Bronnimann and Goodrich result, for general geometric objects we can achieve a hitting set of size $O(log \hspace{1mm} OPT)$ of
the optimal solution.  \par
 
In this report, the second problem that we will be discussing is a geometric hitting set problem where, $X$ is a ground set of points in the plane and $\cal S$ is a 
set of axis parallel rectangles \cite{aronov2010small}. Here we shall show the existence of $\epsilon$ -nets of size $O(1/\epsilon \hspace{1mm} log\hspace{1mm}log\hspace{1mm} 1/\epsilon)$.
Then applying Bronnimann and Goodrich result \cite{bronnimann1995almost} we get the hitting set of size $O(log\hspace{1mm}log \hspace{1mm} OPT)$ of the optimal solution.
This can be extended to axis-parallel boxes as ranges in 3-dimension, leading to the same result as in 2-dimension. In this report we shall just 
focus on the 2 dimensional result.

\section{Hard Capacitated Set Cover}

Recalling the problem statement; we are given a set system $(X,\cal{S})$, where $X = \{x_1,x_2,..., x_n\}$ is the ground set of $n$ elements
and $\cal S$ is the set of subsets of $X$. Each set $S \in \cal S$ has a cost $w(S) > 0$ and a positive integral capacity $k(S) > 0$ associated with it.
Let $\cal R$ be a subset of $\cal S$. Let $f: X \rightarrow \cal R$ be an assignment of elements in $X$ to a set in $\mathcal{R}$ such that for any $x\in X$, if 
(a) $f(x) = S$ then $x \in S$ (b) $|\{x\mid f(x) = S\}| \leq k(S)$ for all $S\in \cal R$. We call $\cal R$ a valid set cover if such $f$ exists.
Cost of the solution $\cal R$ is the sum of the cost of all sets in $\cal R$. The goal is to find a set cover $\cal R$ of minimum cost.
Give an instance $(X,\cal S)$ of this problem, we first need to check whether there exists a valid set cover.
So in this section, firstly we will set up a directed flow network for our instance of hard capacitated set cover problem. Using maximum-flow, 
we will check in polynomial time, whether there exists a valid set cover for the given problem instance. Secondly, we will describe Wolsey's algorithm and finally, 
we will proceed with our analysis to give a set cover of cost $O(log \hspace{1mm} n)$ of the optimal solution.
\subsection{Setting up directed flow network}
Flow networks have played an important role in design of algorithms and in other areas of computer science like network design,
computer vision, etc. Many of the covering problems are visualized as flow network problems \cite{chang2013min}, since there is a rich literature 
of network flow problems available.
In capacitated set cover problem with hard capacities there is a ground set $X = \{1,..,n\}$ of elements and a collection of its subset $\cal S$.
Each $S\in \cal S$ has a capacity $k(S)$ and 
a non-negative cost $w(S)$ associated with it. Let $\cal P$ be a  subset of sets of $\cal S$. Denote $\cal C$ $\subseteq \cal P$ $ \times X$ as a partial cover iff for 
each $(S,e) \in \cal C$, $e \in S$. We say $e\in X$ is covered by $S$ in $\cal C$ if $(S,e) \in \cal C$.
Let us give an example, $\cal P$ $=\{a, b, c\}$, where $a = \{1, 3, 4\}$, 
$b = \{2, 5, 6\}$, $c = \{3, 8, 9, 5\}$ and $X=\{ 1, 2, 4, 8\}$. If $\cal C$ $=\{(a,1), (a,8), (c,8)\}$ then it is not a valid partial cover, since $(a,8) \in \cal C$ but
$8 \notin a$. Where as $\cal C$ $=\{(a,1), (a,4), (c,8)\}$ is a valid partial cover.\par

Assume without loss of generality that each element $e\in X$ is covered at most once in $\cal C$, as we can view covering as assigning element 
$e \in X$ to set $(S,e) \in \cal C$.
Partial cover $\cal C$ is \textit{feasible} if all sets maintain their capacity constraints. Let \textit{value} of $\cal C$ be the number of elements it covers, 
denoted by $|\cal C|$. Let $f(\cal P)$ denote the maximal value of the partial cover $\cal C$ over all the feasible partial covers $\cal C$.
We show that a feasible partial cover $\cal C$ $\subseteq \cal P$ $\times X$ of value $f(\cal P)$, can be computed in polynomial time. 

\begin{lemma}
Given an instance of hard capacitated set cover problem and a subset $\cal P$ of sets of $\cal S$, a cover $\cal C$ of value $f(P)$ can be computed in polynomial time
i.e. we can establish in polynomial time, whether $\cal P$ is a valid set cover.
\end{lemma}
\begin{proof} 
Let us construct a directed network/graph $G = (L, R, E)$ of $\cal P$ and $X$. Each set $S_i \in \cal P$, for all $1 \leq i \leq |\cal S|$,
represents a vertex in L and each element $e_j \in X$, for all $1 \leq j \leq n$ represents a vertex in R. There is an edge $(S_i, e_j) \in E$
of capacity 1 iff $e_j \in S_i$. Add a source vertex $S$ and an edge $(S, S_i)$ of capacity $k(S_i)$, for all i. Add a target vertex $T$ and 
an edge $(e_j, T)$ of capacity 1, for all j (see Figure 4).\\
Here an edge $(S_i, e_j)$ implies $e_j \in S_i$. So a flow in the network is a feasible partial cover.
Since we are using all the edges between $S_i$ and $e_j$, the value of maximum flow is at least $f(\cal P)$. In fact, value of the maximum flow in the network is exactly equal 
to $f(\cal P)$, since the value of the flow is integral. So $\cal P$ is a feasible solution to the set cover problem iff $f(\cal P)$ $= |X|$. 
Since maximum flow value can be computed in polynomial time, so can $f(\cal P)$. This completes our proof.
\end{proof}
\begin{figure}[H]
\centering
\includegraphics[scale = 0.62]{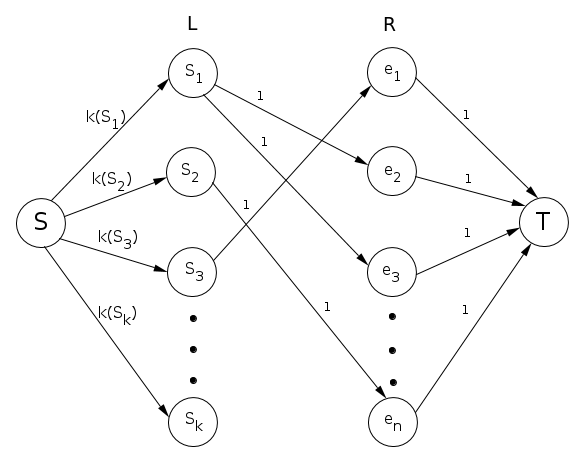}
\caption{\textit{Directed network $G = (L, R, E)$ where;  L: $S_i \in \cal P$, R: $e_j \in X$ and $E$: set of all edges}}
\end{figure}

\subsection{Wolsey's Algorithm and Analysis}

Let $\cal R$ $\subseteq \cal S$ be a family of sets.
The value of $f(\cal R)$ can be computed in polynomial time from Lemma 1. Let $f_\mathcal{R}(S) = f(\cal R$ $\cup \{S\})- f(\cal R)$, i.e. the increase in the number 
of elements covered when set $S$ is added to subset $\cal R$.  

\begin{algorithm}
 \caption{: Wolsey's greedy algorithm }
\begin{algorithmic}
\REQUIRE Feasible capacitated set cover denoted by $\cal P$
\STATE Initially, $\cal P$ $= \phi$.
\WHILE{$\cal P$ is not a feasible capacitated set cover} 
\STATE Let $S= arg \hspace{1mm} min_{S:\hspace{1mm}f_\mathcal{P}{(S)}>0} \hspace{1mm} (w(S)/f_\mathcal{P}{(S)})$  \\
\STATE Add $S$ to $\cal P$ 
\ENDWHILE
\RETURN $\cal P$
\end{algorithmic}
\end{algorithm}

It is important to note that Wolsey's greedy algorithm is quite similar to the weighted set cover greedy algorithm. In greedy algorithm for weighted
uncapacitated set cover (see Algorithm 2), we pick the set in the cover $\cal C$ to be the one that makes most progress i.e. cover the most 
uncovered elements per unit weight. If $X$ is the set of elements that are not covered yet, we add set $S_i$ to the cover $\cal C$, if it minimizes
the quantity $w_i/(S_i \cap X)$.
In this paper \cite{chuzhoy2002covering}, we make a key observation that allows us to use the same analysis as that of the weighted set cover greedy algorithm
\cite{chvatal1979greedy}, for hard capacitated set cover. In the Section 2.3 we have showed this alternative analysis.

\begin{algorithm}
 \caption{: Greedy algorithm for weighted set cover }
\begin{algorithmic}
\REQUIRE minimum set cover $\cal C$ 
\STATE Initially, let ground set of elements be $X = n$ elements, weight $w_i$ associated with each $S_i \in \cal S$ and cover $\mathcal{C} = \phi$; $X\rightarrow Y$
\WHILE{$Y\neq \phi$} 
\STATE Let $S_i$ be the set that minimizes $\dfrac{w_i}{|S_i \cap Y|}$
\STATE $\mathcal{C} = \mathcal{C} \cup \{S_i\}$
\STATE $Y = Y \setminus S_i$
\ENDWHILE
\RETURN $\cal C$
\end{algorithmic}
\end{algorithm}

Let $\cal C$ $\subseteq \cal R$ $\times X$ be a feasible partial cover. For $\cal R'$ $\subseteq \cal R$, let $f_\mathcal{C}(\cal R')$ denote the number of elements 
covered by sets of $\cal R'$ in cover $\cal C$.

\begin{lemma}
 Let $\cal R$ be a valid set cover for an instance of capacitated set cover problem with hard capacities and $\mathcal{R}_1$ $\mathcal{R}_2$ be its partition into
 2 disjoint subsets. There exists a feasible cover $\cal C$ $\subseteq \cal R$ $\times X$ such that:
 \begin{enumerate}
  \item all elements of $X$ are covered in $\cal C$.\vspace{-1.5mm}
  \item $f_\mathcal{C}(\mathcal{R}_1)$ $= f(\mathcal{R}_1)$
 \end{enumerate}
\end{lemma}
\begin{proof}
For a feasible cover $\mathcal{C} \subseteq \mathcal{R} \times X$, where each element $e \in X$ is covered by some set $S \in \cal R$; 
let us assume that $f_\mathcal{C}(\mathcal{R}_1) < f(\mathcal{R}_1)$. Now, consider any feasible partial cover $\mathcal{C}'\subseteq \mathcal{R}_1 \times X$ satisfying the 
condition $f_\mathcal{C'}(\mathcal{R}_1) = f(\mathcal{R}_1)$. 
Since all the elements are not covered by cover $\mathcal{C}'$, we shall use a subroutine which will gradually change the cover $\cal C$ by replacing
some of its set assignments, with the assignment in $\mathcal{C}'$ to achieve the desired cover $\cal C$, while maintaining feasibility.
\begin{algorithm}
 \caption{: Subroutine}
\begin{algorithmic}
\REQUIRE Cover $\cal{C}$ such that $f_\mathcal{C}(\mathcal{R}_1) = f(\mathcal{R}_1)$. 
\WHILE{$f_\mathcal{C}(\mathcal{R}_1) < f(\mathcal{R}_1)$} 
\STATE There exists at least one set $S \in \mathcal{R}_1$ that covers more elements in $\mathcal{C}'$ than in $\mathcal{C}$.
\STATE If $e\in X$ is one of that additional element covered by $S$ in $\mathcal{C}'$ but by some other set $S'$ in $\cal C$; 
\STATE remove $(S',e)$ and add $(S,e)$ to $\cal C$. 
\ENDWHILE
\RETURN $\cal C$
\end{algorithmic}
\end{algorithm}
When the subroutine is executed, the pair $(S,e)$ is added to $\cal C$, only if $S$ covers more elements in $\mathcal{C}'$ than in $C$, 
so the capacity constraint is maintained. Hence this procedure is carried out without violating the feasibility of cover $\cal C$.
It is easy to see that, once an assignment is made in the subroutine, it remains there till the subroutine has completed execution.
The maximum number of iterations taken by the subroutine before terminating, is bounded above by $|\mathcal{C}'|$ and after termination we get the 
desired cover $\cal C$.
\end{proof}  

\begin{figure}[H]
\centering 
\includegraphics[scale = 0.60]{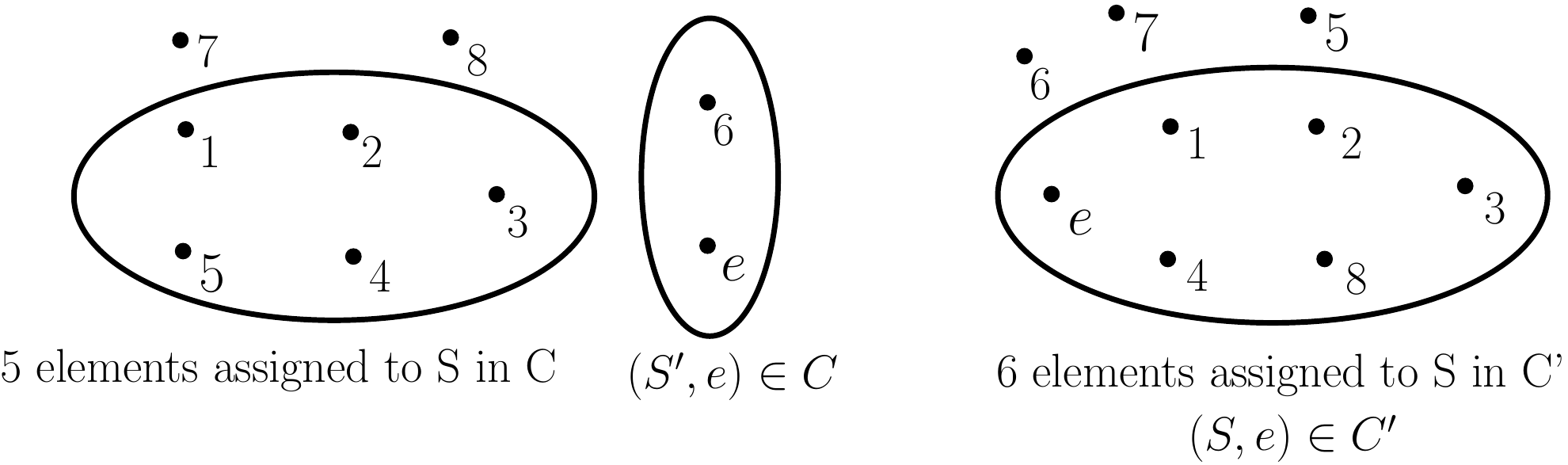}
\caption{\textit{$S$ covers more elements $(e \in E)$ in $\mathcal{C}'$ than in $\cal C$; remove assignment $(S',e)$ and add $(S,e)$ to $\cal C$}}
\end{figure}

\textit{Analysis of Wolsey's Greedy Algorithm: }Let $\mathcal{P} = \{S_1, S_2,..., S_m\}$ be the solution returned by the algorithm in the above order. 
The solution returned by the algorithm in the $i^{th}$ iteration, where $0 \leq i \leq m$, is $\mathcal{P}_i = \{S_1, S_2,...,S_i\}$. Now we re-run
the algorithm.
Initially, $\mathcal{P}_0 = \phi$ i.e. no set is added to the solution by the greedy algorithm. Now we apply Lemma 2 in this iteration, which computes
a feasible cover $\mathcal{C} \subseteq (\mathcal{P}_0 \cup OPT) \times X$ such that each element in $X$ is assigned to exactly one set in OPT.
For each set $S \in OPT$, we define $a_0(S)$ as the number of elements covered by $S$ in $OPT$. 
In the $i^{th}$ iteration $S_i$ is added to the solution by the greedy algorithm i.e. $ P_i = \{S_1, S_2,.., S_i\}$. Now we apply Lemma 2 in this iteration, 
which computes a feasible cover $\mathcal{C} \subseteq (\mathcal{P}_i \cup OPT) \times X$. 
For each $S \in OPT\setminus \mathcal{P}_i$, we define $a_i(S)$ to be the number of elements assigned to a set $S$ in $\mathcal{C}$ where, each set $S$ has capacity of at most
the number of elements assigned to that set in $\mathcal{C} \subseteq (\mathcal{P}_{i-1} \cup OPT) \times X$ i.e. $a_{i-1}(S)$. 
For the rest of the analysis the following statement remains constant and we show this inductively: \textit{all the elements can be covered by 
the sets in $OPT \cup \mathcal{P}_i$, even when the sets $S \in OPT\setminus \mathcal{P}_i$ has a capacity constraint, of at most $a_i(S)$.}\par
Here base case holds for $\mathcal{P}_0$ and $a_0$. Assume that, all the elements can be covered by the sets in $OPT \cup \mathcal{P}_{i-1}$, even when the sets
$S \in OPT\setminus \mathcal{P}_{i-1}$ has a capacity constraint of at most $a_{i-1}(S)$. Since $OPT\setminus \mathcal{P}_{i} \subseteq OPT\setminus \mathcal{P}_{i-1}$, 
it is safe to say that $OPT \cup \mathcal{P}_{i}$ is a feasible cover, even when the sets $S \in OPT\setminus \mathcal{P}_{i}$ has a capacity constraint of at most 
$a_{i-1}(S)$. From Lemma 2, there exists a feasible cover $\mathcal{C} \subseteq (OPT \cup \mathcal{P}_i) \times X $, where $f_\mathcal{C}(\mathcal{P}_i) = f(\mathcal{P}_i)$ 
and each $S \in OPT\setminus \mathcal{P}_i$ covers at most $a_{i-1}(S)$ elements. Since a set $S$ covers same or more number of elements in $OPT\setminus \mathcal{P}_{i-1}$ 
than in $OPT\setminus \mathcal{P}_{i}$, so it can cover same or more number of elements in $\cal C$ when $S \in OPT\setminus \mathcal{P}_{i-1}$ as to when
$S \in OPT\setminus \mathcal{P}_{i}$. 
Hence $a_{i-1}(S) \geq a_{i}(S)$.\par
\textit{Charging: }Let $S_i$ be the set added to the solution by the greedy algorithm in the $i^{th}$ iteration. If the added set $S_i \in OPT$, 
then we don't charge any set in $OPT$, since $OPT$ pays for the same as well. If $S_i \notin OPT$, let the maximum possible increase in the number of
elements that can be covered in $\cal C$, when $S_i$ is added to $\cal P$ after the $(i-1)^{th}$ iteration be denoted by $f_{\mathcal{P}_{i-1}}(S_i) = n_i$. 
Summing over all $S\in OPT\setminus \mathcal{P}_i$, the difference between maximum capacity of a set $S$ i.e. $a_{i-1}(S)$ and the number of elements actually covered 
by $S$ in $\cal C$ i.e. $a_{i}(S)$, gives the maximum possible increase in $\cal C$ on adding $S_i$ i.e. $n_i$. Mathematically;
\begin{equation}
 \sum_{S \in OPT\setminus \mathcal{P}_i}[a_{i-1}(S) - a_{i}(S)] = n_i 
\end{equation}
We charge a cost of $(w(S_i)/n_i)\times[a_{i-1}(S) - a_{i}(S)]$ to each $S \in OPT\setminus \mathcal{P}_i$. Using equation 1, we can see that in $i^{th}$ iteration 
the total cost charged to $S_{i} \in OPT$ is precisely $w(S_i)$.\par   
\textit{Bounding the cost charged to each set in $OPT$: }If the set $S\in OPT$ is added to the solution $\cal P$ by the greedy algorithm, let $(j+1)$ be 
the iteration at which this addition occurs. If the set $S \in OPT$ is not included in $\cal P$, let $(j+1)$ be the first iteration at which $a_j(S) = 0$. 
For each of the iterations $i < j$, the maximum possible increase in the number of elements covered in $\cal C$ at the beginning of the $i^{th}$ iteration
is more than the maximum capacity of $S$ i.e. $f_{\mathcal{P}_{i-1}}(S) \geq a_{i-1}(S)$, otherwise $i$ would have been the last iteration and not $j$.
Since $S \notin \mathcal{P}$, i.e. not selected by the greedy algorithm in this iteration, hence 
$\dfrac{w(S_i)}{n_i} \leq \dfrac{w(S)}{f_{\mathcal{P}_{i-1}}(S)} \leq \dfrac{w(S)}{a_{i-1}(S)}$. 
Here the first inequality follows from the condition in the greedy algorithm and the second follows from above. Total cost charged to $S$ is

\begin{equation}
  \sum_{i=1}^j [a_{i-1}(S) - a_i(S)] \dfrac{w(S_i)}{n_i} \leq w(S)\sum_{i=1}^j\dfrac{a_{i-1}(S) - a_i(S)}{a_{i-1}(S)}
 \end{equation}
 $\forall$ $i$, $1 \leq i\leq j$; add $\dfrac{1}{a_{i-1}(S)}$ for $[a_{i-1}(S) - a_i(S)]$ times and $l\leq a_{i-1}(S)$ implies $\dfrac{1}{l} \geq \dfrac{1}{a_{i-1}(S)}$ 
\begin{equation}
 \sum_{l=a_i(S)+1}^{a_{i-1}(S)} \dfrac{1}{a_{i-1}(S)} \leq \sum_{l=a_i(S)+1}^{a_{i-1}(S)} \dfrac{1}{l}.
\end{equation}
Substituting (3) in (2) we see that, the value charged to $S$ is bounded above by $w(S)H(S)$. Hence, total cost of the solution is $[OPT + OPT.ln(max_s|S|)]$.

\begin{theorem}
 There exists a greedy algorithm for set cover problem with hard capacities, that gives a solution of cost $O(log \hspace{1mm} n)$ of that of the optimal solution.  
\end{theorem}

\subsection{Alternative Analysis for Wolsey's Greedy Algorithm}

In this paper \cite{chuzhoy2002covering}, we make a key observation in their analysis of Wolsey's algorithm for hard capacitated set cover problem, 
that in each iteration of the algorithm there exist an ordering in which the elements of the ground set $X$ are covered. This observation allows us to apply the same
analysis as the weighted set cover greedy algorithm \cite{chvatal1979greedy} to hard capacitated set cover. The observation is stated in the following lemma.

\begin{lemma}
Consider an instance of the set cover problem with hard capacities and let $\mathcal{P}$ be a valid set cover 
and $\mathcal{C} \subseteq \mathcal{R} \times X$ be a feasible cover for this instance. There exists an ordering $\{x_1,x_2,...,x_n\}$ in 
which the greedy algorithm covers elements in the ground set $X$, such that for all $i>0$,\\
(a) the union of the sets in $P_i$ covers $f_{\mathcal{C}}(P_i) = f(P_i)$ elements and\\
(b) they are indexed based on the order they are covered i.e. $\{x_1, x_2,.., x_{f_{\mathcal{P}_i}}\}$. 

\end{lemma}
\begin{proof}
  Let $\mathcal{P} = \{S_1, S_2,..., S_m\}$ be the solution returned by the algorithm in the above order. 
  The solution returned by the algorithm in the $i^{th}$ iteration, where $0 \leq i \leq m$, is $\mathcal{P}_i = \{S_1, S_2,...,S_i\}$.
  Let $\mathcal{C} \subseteq \mathcal{P} \times X$ be a feasible cover where each element is covered by some set $S \in \mathcal{P}$. 
  Without loss of generality assume that no element is covered by more than one set in $\mathcal{P}$. 
  Now we shall prove the lemma using an inductive argument. 
  In the first iteration of the greedy algorithm when the set $S_1$ is added to the solution, denoted by $\mathcal{P}_1$, it covers 
  $f_{\mathcal{C}}(\mathcal{P}_1) = f{(\mathcal{P}_1)}$ elements and they are indexed based on the order they are covered i.e. 
  $\{x_1, x_2,.., x_{f_{\mathcal{P}_1}}\}$.
  Assume that in the $i^{th}$ iteration of the greedy algorithm, when the set $S_i$ is added to the solution denoted by $\mathcal{P}_i$, 
  it covers $f_{\mathcal{C}}(\mathcal{P}_i) = f{(\mathcal{P}_i)}$ elements and they are indexed based on the order they are covered i.e. 
  $\{x_1, x_2,..., x_{f_{\mathcal{P}_{i-1}}}, x_{f_{\mathcal{P}_{(i-1)} + 1}}, x_{f_{\mathcal{P}_{(i-1)} + 2}},...,x_{f_{\mathcal{P}_{i}}}\}$.
  Now let us prove the induction hypothesis that, in the $(i+1)^{th}$ iteration of the greedy algorithm, when the set $S_{i+1}$ is added to the 
  solution, denoted by $\mathcal{P}_{i+1}$, it covers $f_{\mathcal{C}}(\mathcal{P}_{i+1}) = f{(\mathcal{P}_{i+1})}$ elements and they are indexed
  based on the order they are covered i.e. $\{x_1, x_2,.., x_{f_{\mathcal{P}_{i}}},.., x_{f_{\mathcal{P}_{i+1}}}\}$ elements. 
  For proving this, assume that $f_{\mathcal{C}}(\mathcal{P}_{i+1}) < f{(\mathcal{P}_{i+1})}$. Let $\mathcal{C}' \subseteq \mathcal{P} \times X$ 
  be a feasible partial cover, such that $f_{\mathcal{C}'}(\mathcal{P}_{i+1}) = f{(\mathcal{P}_{i+1})}$. Some elements may not be covered by 
  $\mathcal{C}'$, since it is a partial cover. While maintaining the feasibility, we will gradually change the cover $\mathcal{C}$
  using the subroutine given below, until the lemma is satisfied. 
\begin{algorithm}
 \caption{: Subroutine}
\begin{algorithmic}
\REQUIRE Cover $\cal{C}$ such that $f_\mathcal{C}(\mathcal{P}_{i+1}) = f(\mathcal{P}_{i+1})$. 
\WHILE{$f_\mathcal{C}(\mathcal{P}_{i+1}) < f(\mathcal{P}_{i+1})$} 
\STATE There exists at least one set $S \in \mathcal{P}_{i+1}$ that covers more elements in $\mathcal{C}'$ than in $\mathcal{C}$.
\STATE If $e\in X$ is one of that additional element covered by $S$ in $\mathcal{C}'$ but by some other set $S'$ in $\cal C$; 
\STATE remove $(S',e)$ and add $(S,e)$ to $\cal C$. 
\STATE If $e\in X$ is not assigned to any set in $\mathcal{C}$ but assigned to some set $S''\in C'$; add $(S'',e)$
\ENDWHILE
\RETURN $\cal C$
\end{algorithmic}
\end{algorithm}

In the above sub-routine we can see that, for all $i>0$, any element that is already covered by assigning it to a set in the $i^{th}$ iteration, 
remains covered in the $(i+1)^{th}$ iteration as well. Only its assignment to a particular set may change.
The union of the sets in $\mathcal{P}_{i+1}$ cover at least the elements covered by union of the sets in $\mathcal{P}_{i}$ i.e. $f({\mathcal{P}_i})$ elements. 
Any element $e\in X$ that is covered by the union of the sets in $\mathcal{P}_{i+1}$ is (a) either covered in the previous iteration i.e. by union of the sets in
$\mathcal{P}_{i}$, in which case, from the inductive step we can say that, it has already been indexed within $\{x_1, x_2,.., x_{f_{\mathcal{P}_{i}}}\}$ 
(b) or the element $e$ is newly covered in $(i+1)^{th}$ iteration and it is indexed based on the order it is covered i.e. 
$\{x_{f_{(\mathcal{P}_{i})}+1}, x_{f_{(\mathcal{P}_{i})}+2},..., x_{f_{(\mathcal{P}_{i+1})}}\}$. This completes the inductive argument.
Once an assignment is made in the subroutine, it remains there till the subroutine has completed execution.
The maximum number of iterations taken by the subroutine before terminating, is bounded above by $|\mathcal{C}'|$ and after termination we get $P_i$, 
for all $i>0$ that satisfies the lemma. 
\end{proof}
\par
\textit{Alternative Analysis of Wolsey's Greedy Algorithm:} Let every element in $X$ be covered by exactly one set in optimal set cover denoted by OPT. 
 When the greedy algorithm chooses a set $S_{i}$ in the $i^{th}$ iteration, let it charge the price per element for that iteration, to each of the newly 
 covered elements i.e.$f(P_{i}) - f(P_{i-1})$. Total weight of the sets chosen by the algorithm equals the total amount charged. Each element is charged once.
 Consider a set $S' \in OPT$ having $k < n$ elements from $X$, assigned to it in the optimal cover. Let the greedy algorithm cover the elements of $S'$ in the order: 
 $y_k, y_{k-1},.., y_1$. At the beginning of an iteration in which the algorithm covers element $y_i$ of $S'$, at least $i-1$ elements of $S'$ 
 remains uncovered. If the greedy algorithm chooses the set $S'$ to be added to the solution in this iteration, it would pay a cost per element of 
 at most $w_{S'}/i-1$. Thus in this iteration greedy algorithm pays at most $w_{S'}/i-1$ per element covered. Hence element $y_i$ pays at most 
 $w_{S'}/i-1$ to be covered. Now, summing over all the elements in $S'$, the total amount charged to the elements in $S'$ is at most
 $\sum_{i=1}^{k}w_{S'}/i-1$ i.e. $w(S')H_{k-1}$. Now we sum over all the sets $S'\in OPT$ and noting the fact that every element 
 in $X$ is in some set in OPT, the total amount charged to elements overall is at most $\sum_{S' \in OPT}w_S'H_{k-1} = H_n.OPT $  

\begin{theorem}
  For hard capacitated set cover problem, the Wolsey's greedy algorithm returns a set cover of cost at most $H_n$ times of that of the optimal set cover.
\end{theorem}

\section{Geometric Set Cover: Axis-Parallel Rectangles}

In this section, we explore a geometric hitting set problem where, $X$ is a ground set of points in a plane and $\cal S$ is a 
set of axis-parallel rectangles and the goal is to select a minimum cardinality subfamily $X'\subseteq X$, such that every axis-parallel rectangle in 
$\cal S$ contains one point from $X'$. 
We will firstly show the existence of small size $\epsilon$ -nets of size $O(1/\epsilon \hspace{1mm} log\hspace{1mm}log\hspace{1mm} 1/\epsilon)$ for
axis-parallel rectangles. We shall show this in two sub parts (1) Constructing $\epsilon$ -nets (2) Estimating the size of the $\epsilon$ -nets. 
Then we will use the Bronnimann and Goodrich result to show that, if the existence of the above sized net can be shown in polynomial time, then
solution of size $O(log\hspace{1mm}log \hspace{1mm} OPT)$ of the optimal solution, for the corresponding hitting set problem can be computed in polynomial time.
We conclude this section by giving the lower bound example and a key observation, that motivated our technique. 
\subsection{Constructing $\epsilon$ -nets}
Let us start by constructing a balanced binary tree (BBT) $\cal T$, over set $P$ of $n$ points in the plane. We build the tree based on the x-order of these $n$ points and we terminate 
the construction when every leaf node of the tree reaches to a size between the range $[n/r, n/2r]$, where $r = 2/\epsilon$. By size here we mean the number of points under
each leaf node. Since its a BBT, the number of nodes at some level $i$ will be $n/2^{i}$ and due to the termination condition we can say that, the maximum 
level $i_{max}$ of the tree is $n/2r = n/2^{i_{max}}$ i.e. $(1 + log\hspace{1mm}r)$ levels.
\begin{figure}[H]
\centering 
\includegraphics[scale = 0.60]{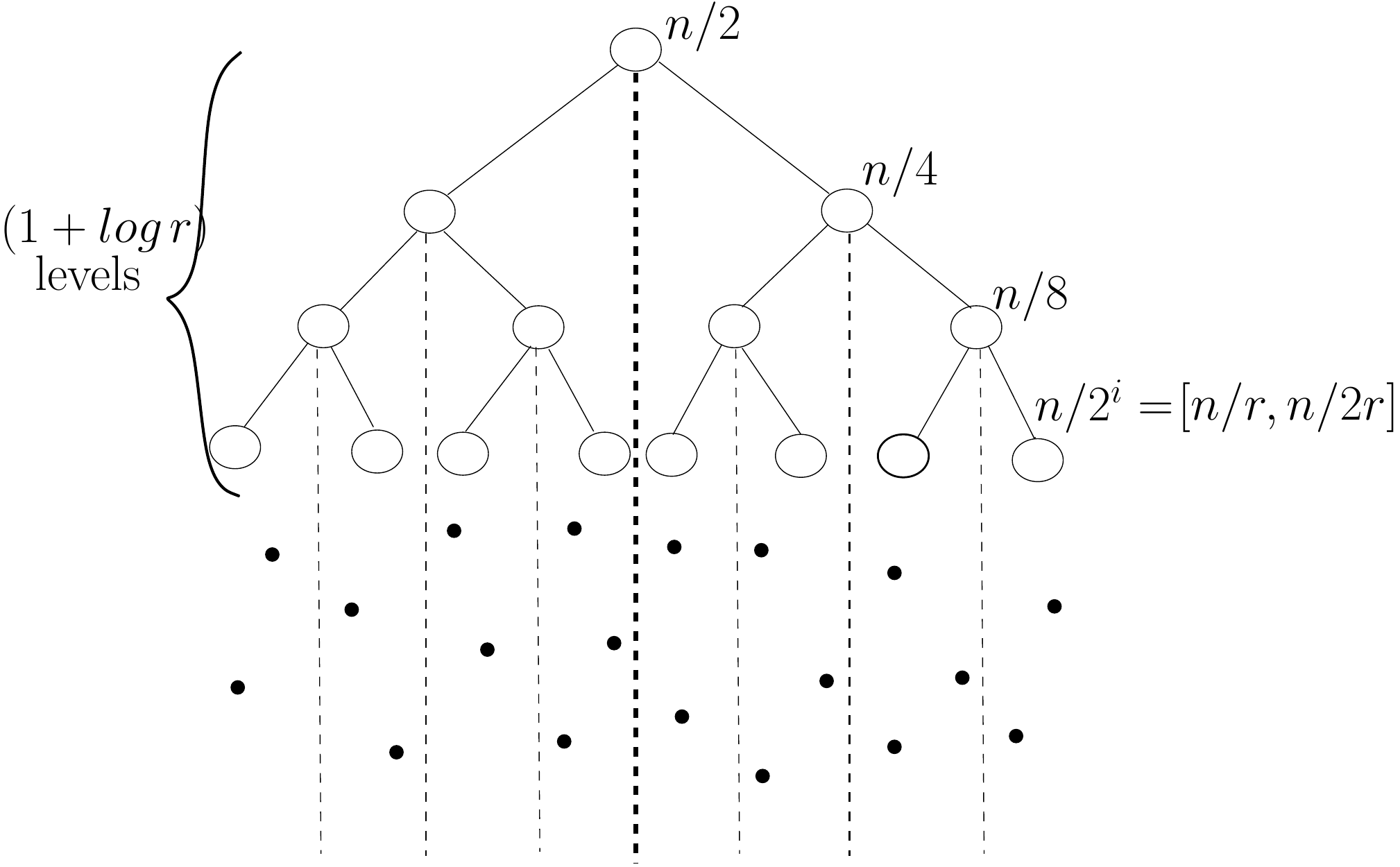}
\caption{\textit{Balanced binary tree, decomposing the point set $P$}}
\end{figure}

We randomly select  a set $R \subseteq P$, which will be included in the $\epsilon$ -net we construct, such that each point in $P$ that is included in $R$ is selected independently 
with probability $\pi = s/n$, thus $E[|R|] = s$. Set $s = cr\hspace{1mm}log\hspace{1mm}log\hspace{1mm}r$. In figure (7a), for each node $v$, let $P_v$ be the subset of the points in the 
sub-tree rooted at $v$. Let $l_v$ be the vertical line that divides $P_v$ into two subsets $P_{v_1}$ and $P_{v_2}$, which contains points in the sub-tree rooted at $v_1$ 
and $v_2$ respectively. Let the two half planes corresponding to $P_{v_1}$ and $P_{v_2}$ be denoted as $\sigma_{v_1}$ and $\sigma_{v_2}$. Let the plane which contains 
all the points in the sub-tree rooted at the root node $v$, be denoted as $\sigma_{root}$. Our sampling of set $R$ is unbiased, since $P_v$ is set before selecting $R$. \par
In figure (7b), $\sigma_{v}$ is the right (resp. left) portion of $\sigma_{u}$ bounded by $l_u$ on the left hand side (resp. right). Hence we define $l_u$ as the 
\textit{left entry side} (resp. \textit{right entry side}) of $\sigma_{v}$. Let $\overline Q$ be an axis-parallel rectangle that contains at least $\epsilon n$ points of $P$.
Let $u$ be the highest node in the $\cal T$ such that, the vertical line 
assigned to it cuts $\overline Q$ into 2 parts. Let $Q$ be one of the 2 parts, which has at least $\epsilon n/2 = n/r$ points from $P$. There has to exist one such partition.
Let $v$ be the child of $u$ that contains $Q \subset \sigma_v$, hence $Q$ is said to be \textit{anchored} at left entry side $l_u$ of $\sigma_v$.
For the rest of the section, without loss of generality we will assume that the entry side $l_u$ of $\sigma_v$ is the left side. 

\begin{figure}[H]
\centering 
\includegraphics[scale = 0.65]{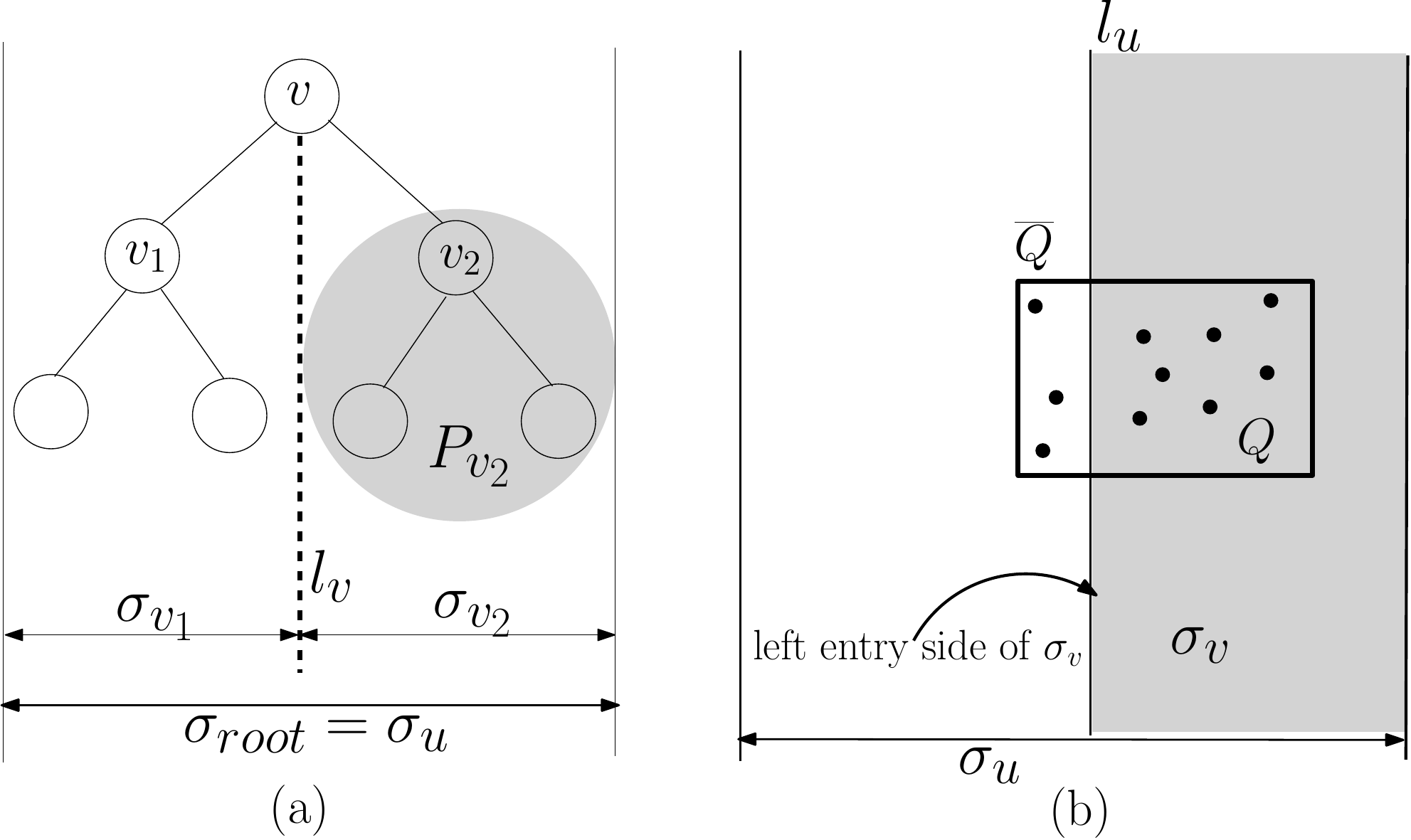}
\vspace{-3mm}
\caption{\textit{(a) $l_v$ splits $P_{v}$ into 2 subsets $P_{v_1}$ and $P_{v_2}$ in half plane $\sigma_{v_1}$ and $\sigma_{v_2}$ respectively 
(b) Half rectangle Q anchored at left entry side $l_u$ of strip $\sigma_v$}}
\end{figure}
\vspace{-1.5mm}
We want to stab all the heavy rectangles i.e. we want to construct a subset from $P$, that intersects every rectangle $\overline Q$ such that $|\overline Q| \geq \epsilon n$.
Since $Q \subset \overline Q$, if any point from our first random sample $R \subseteq P$ is contained in $Q$ we have stabbed $\overline Q$ successfully and we are done.
Let us assume that $Q$ does not contain a point from $R$, and hence we will say $Q$ is \textit{$R$-empty} or \textit{$R_v$-empty}. We define a set $\mathcal{M}_v$ for each node
$v \in \cal T$, that contains all the maximal anchored $R_v$-empty axis-parallel rectangles in $\sigma_v$. We shall then bound the size of $\mathcal{M}_v$. Note that for any $v\in \cal T$, 
$R_v$ is the set of points from our random subset $R$ that is contained in $\mathcal{M}_v$. 

\begin{figure}[H]
\centering 
\includegraphics[scale = 0.79]{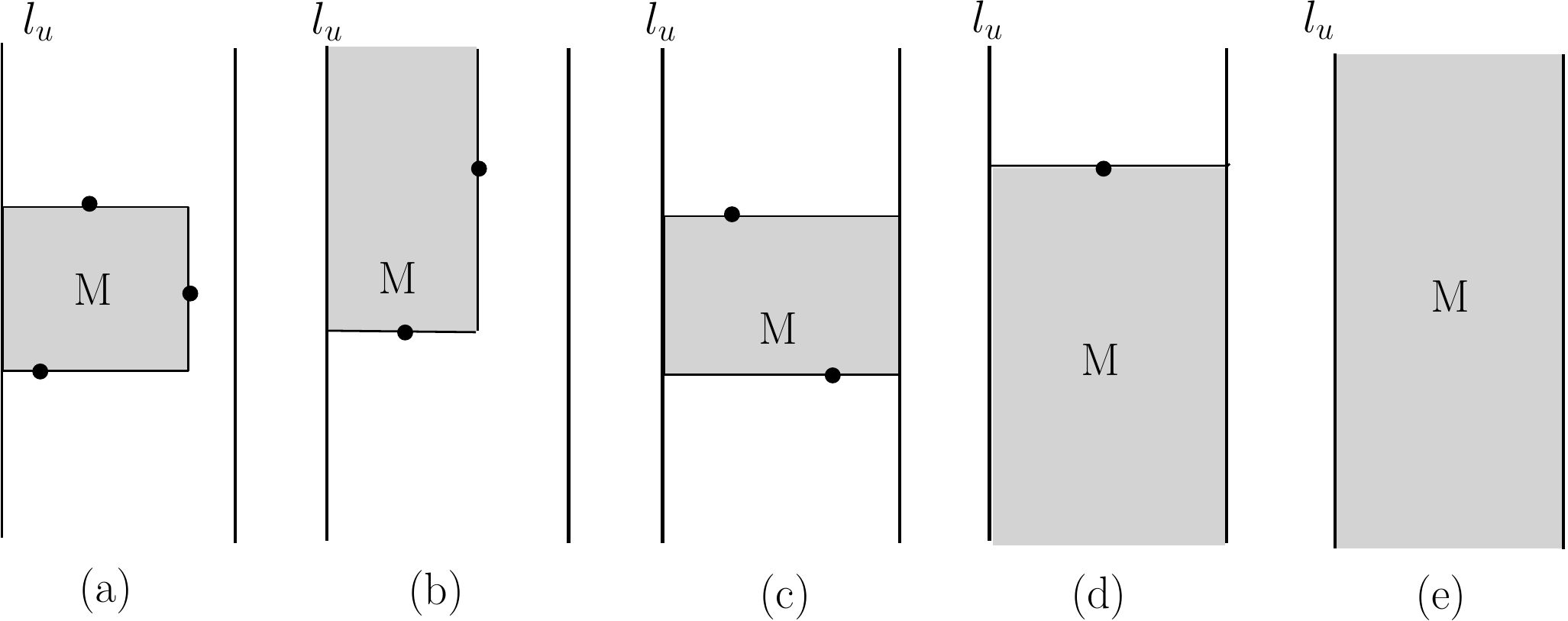}
\vspace{-3mm}
\caption{\textit{(a) Rectangle $M \in \mathcal{M}_v$ is defined by 3 points from $R_v$, one point on each of the 3 unanchored sides, $\mathcal{M}_v$ also contain degenerate rectangles with some
or all of these 3 points missing,
(b) and (c) $M$ with 2 points on the two unanchored sides each (its easy to see the missing case similar to case (b)), (d) $M$ with just one point on the unanchored side 
(either top-most or bottom-most), (e) $M$ with no points from $R_v$, in this case you just have the entire strip as one maximal $R_v$-empty rectangle.}}
\end{figure}

\begin{lemma}
 For each node $v \in \cal T$, the number of rectangles in the set $\mathcal{M}_v$ is equal to $(2r_v + 1)$, where $r_v = |R_v|$.
\end{lemma}
\begin{proof}
We shall cover all of the following, possible cases to show the required bound (see Figure 9). In case (a), when $q\in R_v$ lies on the right unanchored side 
of any rectangle $M \in \mathcal{M}_v$ then, no other rectangle in $\mathcal{M}_v$ can have $q$ on its right unanchored side. Let us argue this by contradiction. Assume that some other 
rectangle $M'\in \mathcal{M}_v$ also has $q$ on its right unanchored side, this will lead to $M' \subset M$, hence it is not a maximal rectangle which contradicts the definition of 
$\mathcal{M}_v$ else, $M'$ contains a point that belongs to $R_v$ which contradicts our assumption that $Q$ does not contain any point from $R$. Hence maximum number of such rectangles
is $r_v$.\par
In case (b), the points $q_1, q_2 \in R_v$ are on the top and bottom unanchored sides of rectangle $M$, then these points are consecutive in $R_v$ in y-coordinate, if they are not then
$M$ will to contain the points which lies between these two points in $R_v$ based on y-order, which is a contradiction. There are $r_v -1$ such pairs.
Finally, we are left with case (c), where there are 2 semi-unbounded rectangles, one bounded from above and one from below by their highest and lowest points respectively.
Thus from the above 3 cases we get our required bound of $|\mathcal{M}_v| = (2r_v + 1)$. It is also easy to see that, the bound holds when $r_v = 0$ where, $M$ is the entire strip and 
$r_v = 1$ where, we get $M$ as the entire plane and 2 semi-unbounded rectangles.  \par 
We are left to show that, $Q$ belongs to at least one rectangle in $\mathcal{M}_v$. Since we have assumed that the entry side of $\sigma_v$ is the left side, we expand $Q$ 
by pushing its right unanchored side till it touches a point of $R_v$ or reaches the boundary of $\sigma_v$. Then we extend it vertically in both the directions till it 
reaches a point in $R_v$ or we let it extend till $\pm \infty$. Thus, the rectangle we get after expansion contains $Q$ and also belongs to $\mathcal{M}_v$. This completes the proof.
\end{proof}

\begin{figure}[H]
\centering 
\includegraphics[scale = 0.82]{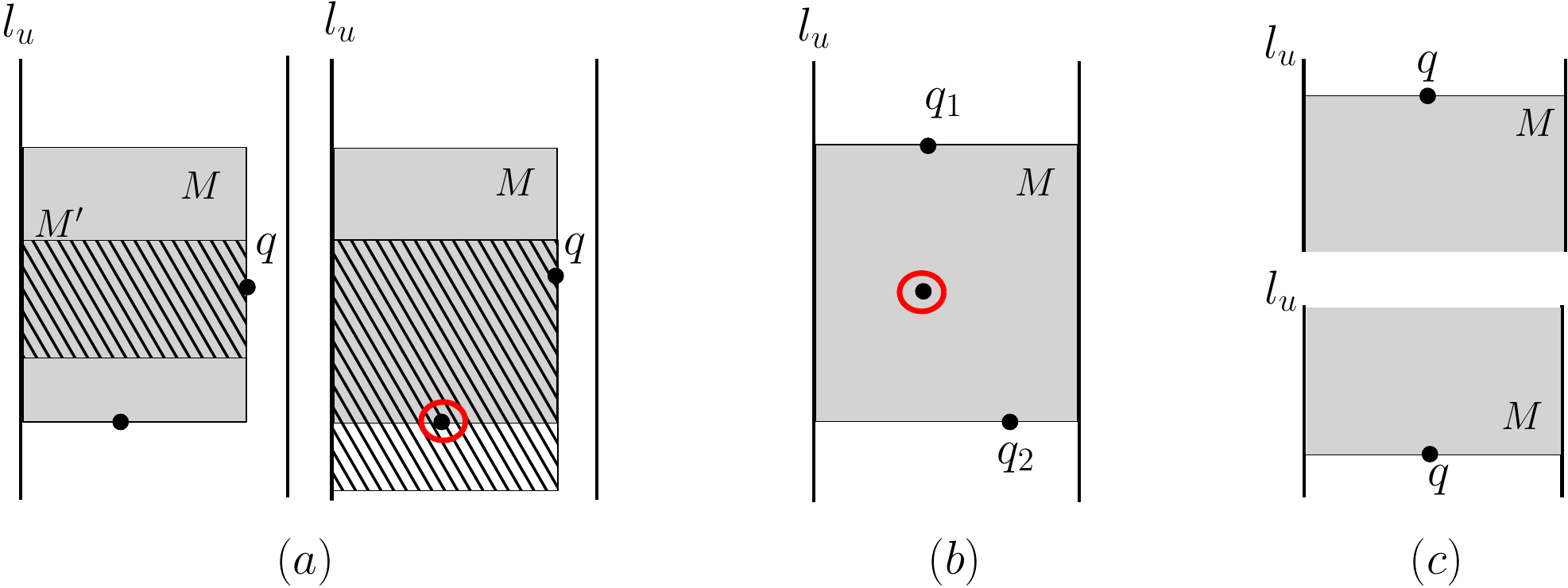}
\vspace{-1mm}
\caption{\textit{Contradicting cases for (a) $q\in R_v$ lies on the right unanchored side of exactly one $M \in \mathcal{M}_v$ 
(b) $q_1, q_2 \in R_v$ are on the top and bottom unanchored sides of $M \in \mathcal{M}_v$, then these points are consecutive in $R_v$ in y-coordinate, and 
(c) 2 semi-unbounded rectangles $M$}}
\end{figure}
For all nodes $v \in \cal T$, overall number of rectangles $M \in \mathcal{M}_v$ at a fixed level is $O(|R| + r')$, where $r'$ is the total number of nodes at a fixed level
and over all the levels of $\cal T$ is $O(|R|log\hspace{1mm}r + r'log\hspace{1mm}r)$ i.e. $O(|R|log\hspace{1mm}r + r)$.
For each $v \in \cal T$ and each $M \in \mathcal{M}_v$, let $t_M$ be the \textit{weight factor} of $M$ and it is defined as $s|M \cap P|/n$. The rectangles $M$ with 
$t_M < s/r = c\hspace{1mm}log\hspace{1mm}log\hspace{1mm}r$ i.e. $s|M \cap P|/n < s/r$ i.e. $|M \cap P| < n/r$ and these rectangles can be ignored since they
have no anchored rectangles $Q$ contained in them. This is because $|Q| \geq \epsilon n/2 = n/r$ points of P. We will just consider rectangles with 
$t_M \geq c\hspace{1mm}log\hspace{1mm}log\hspace{1mm}r$.\par
Using Haussler and Welzl result \cite{haussler_epsilon-nets_1987}, for each $M \in \mathcal{M}_v$ having $t_M \geq c\hspace{1mm}log\hspace{1mm}log\hspace{1mm}r$, there exists 
$N_M \subseteq M \cap P_v$ of size $k\hspace{0.5mm}t_M\hspace{1mm} log\hspace{1mm} t_M$, that forms $(1/t_M)$ -net for $M \cap P_v$, where $k$ is a constant.  
We are done with the construction phase of $\epsilon$ -nets. So the final $\epsilon$ -net is union of set $R$ that we picked in the first level of sampling, with 
the sets $N_M$, $\forall$ $v \in \cal T$ and $\forall$  $M \in \mathcal{M}_v$ with $t_M \geq c\hspace{1mm}log\hspace{1mm}log\hspace{1mm}r$, in the second level of sampling.

\begin{lemma}
 The constructed net $N = R \cup \sum_{v \in T} \sum_{M \in \mathcal{M}_v}N_M$ is an $\epsilon$ -net.
\end{lemma}
\begin{proof}
Since our final $\epsilon$ -net $N$ contains all the points in $R$, it is enough for us to show that for any $R-empty$ rectangle $Q$ containing at least $n/r$
points of $P$, for any $M \in \mathcal{M}_v$ containing $Q$; $Q \cap N_M \neq \phi$.
\begin{equation}
 \dfrac{|Q \cap P|}{|M \cap P|} \geq \dfrac{n/r}{nt_M/s} = \dfrac{s}{rt_M} = \dfrac{c\hspace{1mm}log\hspace{1mm}log\hspace{1mm}r}{t_M} \geq \dfrac{1}{t_M}
\end{equation}

The first inequality above comes from the definition of weight factor. The above equation says that, rectangle $Q$ contains at least $(1/t_M)$ fraction of points in $P$
contained in $M$.
The result we got from the Haussler and Welzl paper, that $N_M$ is $(1/t_M)$ -net of size $k\hspace{0.5mm}t_M\hspace{1mm} log\hspace{1mm} t_M$ for $M \cap P$, which means if
we pick any $(1/t_M)$ points from $M \cap P$, it will surely contain a point from set $N_M$ and $|Q \cap P|$ is at least $(1/t_M)$ of $|M \cap P|$. Thus it follows
that $Q \cap N_M \neq \phi$. This completes the proof.
\end{proof}

\subsection{Estimating the size of N}
The expected size of $N$ is 
\begin{equation}
E[|R| + k \sum_{v} \sum_{\substack{{M \in \mathcal{M}_v}\\ {t_M \geq c\hspace{1mm}log\hspace{1mm}log\hspace{1mm}r}}} t_M\hspace{1mm} log\hspace{1mm} t_M]
\end{equation}
$$ = cr\hspace{1mm}log\hspace{1mm}log\hspace{1mm}r + k \sum_{v} \sum_{\substack{{M \in \mathcal{M}_v}\\ {t_M \geq c\hspace{1mm}log\hspace{1mm}log\hspace{1mm}r}}} t_M\hspace{1mm} log\hspace{1mm} t_M]$$

We shall now fix a level $i$ in the tree $\cal T$ for the rest of the analysis. Since $\cal T$ is a balanced binary tree, for each node $v$ at level $i$, $|P_v| = n/2^i$.
Let the union of the collections of all rectangles $M \in \mathcal{M}_v$, for all the nodes $v$ at level $i$ be denoted by $CT(R)$. Let $CT_t(R) \subseteq CT(R)$ be rectangles
with $t_M \geq t$ where, $t > 0$. Let $R'$ be another random sample from $P$, such that each point from $P$ is chosen in $R'$ independently with probability $\pi'=\pi/t$.
Let $\cal C$ be a set of those rectangles $M$, anchored at the entry side of $\sigma_v$ for all the nodes at level $i$ and has one point of $P$ on each of its 3 unanchored sides.
Note that, all the degenerate rectangles can be handled in analogous manner. For a rectangle $M \in \cal C$, we denote its \textit{defining set} $D(M)$ by the 3 points on its 
unanchored sides and its \textit{killing set} $K(M)$ by the points that lie in the interior of $M$ and belong to $P$. 
We shall now define an axiom, which trivially holds true by our construction and assumption that, no point of $R$ is contained in $Q$ and thus not contained in $\mathcal{M}_v$.

\begin{axiom}
 Every rectangle $M \in \cal C$ belongs to the set $CT(R)$ iff $D(M) \subseteq R$ and $(K(M) \cap R) = \phi$.
\end{axiom}
We need the following exponential decay lemma \cite{Agarwal:1994},\cite{aronov2010small} for our analysis.

\begin{lemma}
 $E[|CT_t(R)|] = O(2^{-t} E[|CT(R')|])$
\end{lemma}
\begin{proof}
Let $\cal Z$ be a collection of all axis-parallel rectangles anchored at the entry side of $\sigma_v$, at a fixed level $i$. Each unanchored sides of
these rectangles has one point of $P_v$ (degenerate cases can be handled the same way). Let $\mathcal{Z}_t \subseteq \cal Z$, contain all the rectangles with $t_M \geq t$. Since 
$CT_t(R) \subseteq CT(R)$ and since we are summing over all the $M \in (\mathcal{Z}_t \subseteq \cal Z)$ we can write,
\begin{equation}
 E[|CT_t(R)|] = \sum_{M \in \mathcal{Z}_t}Pr\{M \in CT(R)\}
\end{equation}
\begin{equation}
E[|CT(R')|] = \sum_{M \in \cal Z}Pr\{M \in CT(R')\} \geq \sum_{M \in \mathcal{Z}_t}Pr\{M \in CT(R')\}
\end{equation}
The above inequality comes from the fact that, $\mathcal{Z}_t \subseteq \cal Z$.
We are done, if we show that for each $M \in \mathcal{Z}_t$; $$\dfrac{Pr\{M \in CT(R)\}}{Pr\{M \in CT(R')\}} = O(2^{-t})$$ \par
Let $A$ be the event that, $D(M) \subseteq R$ and $K(M) \cap R = \phi$ and $A'$ be the event that, $D(M) \subseteq R'$ and $K(M) \cap R' = \phi$. From axiom 1 and the fact that
all the rectangles of $M \in C$, we can say $A$ is the event that $M \in CT(R)$ and $A'$ is the event that $M \in CT(R')$. Let us now denote $\rho = |D(M)| \leq 3$ and $w = |K(M)|$.
A point from point set $P$ is included in $R$ with probability $\pi$ and in $R'$ with probability $\pi'$. Hence $Pr\{A\} = \pi^{\rho}(1-\pi)^w$ and
$Pr\{A'\} = \pi'^{\rho}(1-\pi')^w$. $\pi = s/n$, $\pi' = \pi/t$ and $w\geq n/r = tn/s$. Hence,
$$\dfrac{Pr\{M \in CT(R)\}}{Pr\{M \in CT(R')\}} = \dfrac{Pr\{A\}}{Pr\{A'\}} = \dfrac{\pi^{\rho}(1-\pi)^w}{\pi'^{\rho}(1-\pi')^w} = t^\rho \Big(\dfrac{1-\pi}{1-\pi'}\Big)^w = O(2^{-t})$$
This concludes the proof. 
\end{proof}

We shall now apply Lemma 6, by substituting $t = c\hspace{1mm}log\hspace{1mm}log\hspace{1mm}r$. For all $v$, the collection of all the maximal $R'_v-empty$ rectangles
anchored on their entry side, in their corresponding strips $\sigma_v$, at a fixed level $i$, is denoted by $CT(R')$. Applying Lemma 4, on all the nodes of level $i$ we get, 
\begin{equation}
 E[|CT(R')|] = \sum_{v}(2r'_v + 1)
\end{equation}
where, $R'_v = R' \cap \sigma_v$ and $r'_v = |R'_v|$. Since all the nodes at a particular level of a balanced binary tree are disjoint follows, sets $R'_v$ at level $i$ are
also disjoint. So we get, 
\begin{equation}
 \sum_{v}r'_v = |R'|
\end{equation}
Since we terminate the construction of the balanced binary tree when the size of every leaf node ranges between $[n/r, n/2r]$, we can say that the
maximum number of nodes at level $i_{max}$ ($i_{max}$ is the level which has the maximum nodes) is at most
\begin{equation}
\dfrac{n}{2^{i_{max}}} = \dfrac{n}{2r} \implies 2^{i_{max}} = 2r 
\end{equation}
Substituting equation (9) and (10) in (8) we get 
\begin{equation}
E[|CT(R')|] = \sum_{v}(2r'_v + 1) \leq 2|R'| + \sum_{v} 1 = 2|R'| + 2r = O(r)
\end{equation}
Substituting the value of $E[|CT(R')|]$ in lemma 4, we get
\begin{equation}
 E[|CT_t(R)|] = O(2^{-t} E[|CT(R')|]) = O(2^{-c\hspace{1mm}log\hspace{1mm}log\hspace{1mm}r}r) = O\Big(\dfrac{r}{log^c\hspace{1mm}r}\Big)
\end{equation}
Making it general for any $j\geq t$, we can say $E[|CT_j(R)|] = O(r/2^j)$.  Now we will calculate the number of nodes from a fixed level $i$, that contributes to the size 
of the net $N$. Let us calculate the second part of equation (5).
\begin{equation}
 E[\sum_{v\hspace{1mm}at\hspace{1mm}level\hspace{1mm}i} \sum_{\substack{{M \in \mathcal{M}_v}\\ {t_M \geq t}}} t_M\hspace{1mm} log\hspace{1mm} t_M]
\end{equation}

\begin{align*}
&= E[\sum_{j \geq t} \hspace{1mm} \sum_{\substack{{M \in CT(R)}\\ {t_M = j}}} j\hspace{1mm} log\hspace{1mm} j]\\
&= E[\sum_{j \geq t} j\hspace{1mm} log\hspace{1mm} j(|CT_j(R)| - |CT_{j+1}(R)|)]\\
&= E[(t\hspace{1mm}log \hspace{1mm}t)|CT_t(R)| + \sum_{j > t} (j\hspace{1mm} log\hspace{1mm}j - (j - 1)log(j-1))|CT_j(R)|]\\
&= O\Big(\dfrac{r}{log^{c}\hspace{1mm}r}(t\hspace{1mm}log \hspace{1mm}t) + \sum_{j>t}\dfrac{r}{2^j}log j\Big)\\
&= O\Big(\dfrac{rt\hspace{1mm}log \hspace{1mm}t}{log^{c}\hspace{1mm}r}\Big)\\
&= O\Big(\dfrac{r\hspace{1mm}log\hspace{1mm}log\hspace{1mm}r\hspace{1mm}log\hspace{1mm}log\hspace{1mm}log\hspace{1mm}r}{log^c \hspace{1mm}r}\Big)
\end{align*}

Since this is just for a fixed level $i$, so we repeat this analysis for all the $(1 + log\hspace{1mm}r)$ levels,
$$= O\Big(\dfrac{r\hspace{1mm}log\hspace{1mm}log\hspace{1mm}r\hspace{1mm}log\hspace{1mm}log\hspace{1mm}log\hspace{1mm}r}{log^{c-1} \hspace{1mm}r}\Big)$$
Using the above result in equation (5) we get 
\begin{equation}
 E[|N|] = \Big(cr\hspace{1mm}log\hspace{1mm}log\hspace{1mm}r + \dfrac{r\hspace{1mm}log\hspace{1mm}log\hspace{1mm}r\hspace{1mm}log\hspace{1mm}log\hspace{1mm}log\hspace{1mm}r}{log^{c-1}\hspace{1mm}r}\Big) = O(r\hspace{1mm}log\hspace{1mm}log\hspace{1mm}r)
\end{equation}

\begin{theorem}
 For axis-parallel rectangles, there exists an $\epsilon$ -net ($\epsilon > 0$) of size $O(\dfrac{1}{\epsilon}\hspace{1mm}log\hspace{1mm}log \hspace{1mm} \dfrac{1}{\epsilon})$, for any point set 
 $P$ of n-points in a plane.
\end{theorem}
 Now using the result by Bronnimann and Goodrich \cite{bronnimann1995almost}, we come up with the following theorem,
\begin{theorem}
 For the hitting set problem of axis parallel rectangles, there exists a polynomial time algorithm which achieves a solution of size 
 $O(log\hspace{1mm}log\hspace{1mm}OPT)$ of the optimal solution.  
\end{theorem}

\subsection{Key idea}
It is important to see that, if we make a bad choice of sample $R$ in the first level of sampling, the maximum number of $R-empty$ rectangles can be $\theta(s^2)$. It can be seen in 
the quadratic lower bound example in Figure (10a). Here each point of the lower staircase is matched with its corresponding point in upper staircase. Hence we use 
the tree decomposition technique over the balanced binary tree we constructed, to prune most of those rectangles, such that we are left with $O(s\hspace{1mm}log \hspace{1mm}r)$. 
In Figure (10b), consider the points to the left of $l_u$, we show that the number of maximal $R-empty$ rectangles anchored at $l_u$ is linear in $|R \cap \sigma_v|$,
for all $v \in \cal T$.
 
\begin{figure}[H]
\centering 
\includegraphics[scale = 0.70]{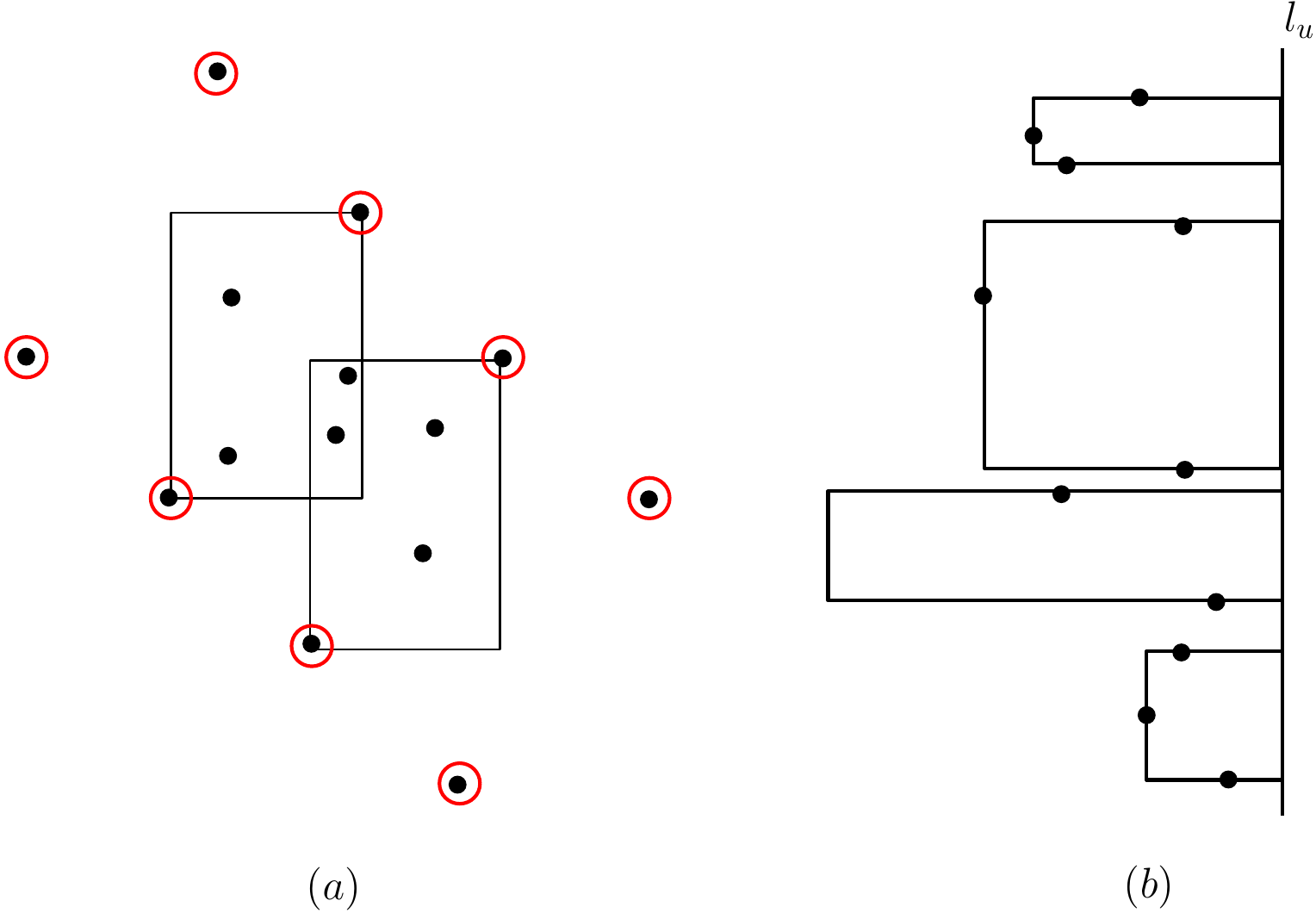}
\vspace{-1mm}
\caption{\textit{(a) Quadratic lower bound example (b) Decomposition of the point set into canonical subsets}}
\end{figure}
\section{Conclusion}
The first part of this report describes the following result that, logarithmic approximation factor for hard
capacitated set cover can be achieved from Wolsey's work \cite{wolsey1982analysis}, using a simpler and more intuitive analysis. We further show in our work, that $O(log \hspace{1mm} n)$ approximation factor can be achieved for the same problem by applying analysis of general set cover to analyze Wolsey's algorithm. This work is based on the key observation that we make in Lemma 3. The second part of the report describes the geometric hitting set problem, where X is a ground set of points in a plane and S is a set of axis parallel rectangles. It is shown that $\epsilon$-nets of size O($1/\epsilon \hspace{1mm} log log \hspace{1mm} 1/\epsilon$) can be computed in polynomial time. Applying Bronnimann and Goodrich result \cite{bronnimann1995almost} gives the hitting set of size $O(log log \hspace{1mm}OPT)$ for this problem.
One open problem is to consider the dual version of the geometric hitting set problem, described in this
report. Namely, given a collection S of n axis parallel rectangles and each rectangle in the subset S contains some point in the plane. We want to show the existence of a small size $\epsilon$-net i.e $\mathcal{S'} \subseteq \mathcal{S}$ whose union contains all the points contained in $\epsilon n$ rectangles of S. It is not known whether the same method \cite{aronov2010small} can also be extended to the dual version. One of the future directions that I am pursuing, is to apply the notion of hard capacities of combinatorial set cover, to the geometric set cover problem.

\bibliographystyle{plain}
\bibliography{My}

\end{document}